\documentclass[sigconf]{acmart}
\renewcommand\footnotetextcopyrightpermission[1]{} 
\usepackage[ruled,vlined, linesnumbered]{algorithm2e}
\newtheorem{theorem}{Theorem}[section]

\newtheorem{problem}{Problem}
\usepackage{soul}

\AtBeginDocument{%
  \providecommand\BibTeX{{%
    \normalfont B\kern-0.5em{\scshape i\kern-0.25em b}\kern-0.8em\TeX}}}

\setcopyright{acmcopyright}
\copyrightyear{2022}
\acmYear{2022}
\acmDOI{10.1145/1122445.1122456}

\acmConference[ACM SIGMOD/PODS '22]{ACM SIGMOD/PODS '22: ACM Conference on Management of Data}{June 12--17, 2022}{Philadelphia, PA}
\acmBooktitle{Woodstock '18: ACM Symposium on Neural Gaze Detection,
  June 03--05, 2018, Woodstock, NY}
\acmPrice{15.00}
\acmISBN{978-1-4503-XXXX-X/18/06}

\begin{document}

\title{An Efficient Pruning Process with Locality Aware Exploration and Dynamic Graph Editing for Subgraph Matching}

\author{Zite Jiang}
\email{jiangzite19s@ict.ac.cn}
\orcid{}
\affiliation{%
  \institution{University of Chinese Academy of Sciences}
  \state{Beijing}
  \country{China}
}

\author{Boxiao Liu}
\email{liuboxiao@ict.ac.cn}
\orcid{}
\affiliation{%
  \institution{University of Chinese Academy of Sciences}
  \state{Beijing}
  \country{China}
}

\author{Shuai Zhang}
\email{zhangshuai-ams@ict.ac.cn}
\orcid{}
\affiliation{%
  \institution{University of Chinese Academy of Sciences}
  \state{Beijing}
  \country{China}
}

\author{Xingzhong Hou}
\email{houxingzhong@ict.ac.cn}
\orcid{}
\affiliation{%
  \institution{University of Chinese Academy of Sciences}
  \state{Beijing}
  \country{China}
}

\author{Mengting Yuan}
\email{ymt@whu.edu.cn}
\orcid{}
\affiliation{%
  \institution{Wuhan University}
  \state{Wuhan}
  \country{China}
  \postcode{43017-6221}
}

\author{Haihang You}
\email{youhaihang@ict.ac.cn}
\orcid{}
\affiliation{%
  \institution{Institute of Computing Technology, Chinese Academy of Sciences}
  \state{Beijing}
  \country{China}
}

\renewcommand{\shortauthors}{Jiang, et al.}

\begin{abstract}
Subgraph matching is a NP-complete problem that extracts isomorphic embeddings of a query graph $q$ in a data graph $G$. In this paper, we present a framework with three components: $Preprocessing$, $Reordering$ and $Enumeration$. While pruning is the core technique for almost all existing subgraph matching solvers, it mainly eliminates unnecessary enumeration over data graph without alternation of query graph. By formulating a problem: Assignment under Conditional Candidate Set(ACCS), which is proven to be equivalent to Subgraph matching problem, we propose Dynamic Graph Editing(DGE) that is for the first time designed to tailor the query graph to achieve pruning effect and performance acceleration. As a result, we proposed DGEE(Dynamic Graph Editing Enumeration), a novel enumeration algorithm combines Dynamic Graph Editing and Failing Set optimization. Our second contribution is proposing fGQL , an optimized version of GQL algorithm, that is utilized during the $Preprocessing$ phase. Extensive experimental results show that the DGEE-based framework can outperform state-of-the-art subgraph matching algorithms.


\end{abstract}




\maketitle
\setcopyright{none}
\pagestyle{plain}

\section{Introduction}
Nowadays, massive multi-source heterogeneous data is generated and accumulating in contemporary society. Graph is the data structure suitable for depicting such data, and there is an arsenal of graph algorithms to retrieve valuable information. 

Subgraph matching is a core graph algorithm known as an NP-complete problem \cite{hartmanis1982computers}. It has applications in many fields. For chemoinformatics \cite{yan2004graph, yang2007path}, finding similarities in the structural formulas of compounds can be converted to a subgraph matching problem. And in the graph database \cite{hong2015subgraph, kim2018turboflux}, the rules are represented in the form of graphs. When querying for the data of a given rule, it could be transformed into a subgraph matching problem. And subgraph matching can also be applied in social networks \cite{fan2012graph}. Therefore, to improve the performance of subgraph matching and design an efficient algorithm has been a research topic for decades.  

In general, the subgraph matching problem can be described as that given a data graph $G$, to find a subgraph $G'$ that is isomorphic to the query graph $q$.
The research on subgraph matching can be further categorised as finding exact solutions and finding approximate solutions. 
The exact solution solvers can be further categorized into three typical methods: \textit{Enumeration-based}, \textit{joint-based} and \textit{constraint programming}. 
The \textit{Enumeration-based} algorithms usually follow the depth-first search strategy as summarized in \cite{sun2020memory}. And \textit{joint-based} algorithms are designed for small query graphs with a few vertices like RapidMatch \cite{sun2020rapidmatch}.  
As for the \textit{constraint programming}, Glasgow Subgraph Solver \cite{archibald2019sequential} uses constraint programming to tackle a wide range of graphs. For non-exact solution, there are approximation algorithms such as SUBDUE \cite{cook1993substructure} and LAW \cite{wolverton2003law}. And these approximation algorithms use certain features such as distance to measure similarity of graphs \cite{bunke1983inexact, sanfeliu1983distance}. Recent years, graph neural networkGNN) has been utilized to solve subgraph matching problem. For example, neural subgraph isomorphism counting \cite{liu2020neural} was proposed recently. In this paper, we focus on the exact subgraph matching problem, specifically \textit{Enumeration-based} subgraph matching method.

To design an efficient subgraph matching algorithm, it is inevitable to prune the infeasible vertices along the enumeration. There exists two strategies: static pruning and dynamic pruning. 




Based on substantial review of previous research, we identify two challenges that we will try to give solutions in this paper. 
It is common practice to utilize connectivity of query graph to optimize enumeration process and carry out pruning which has been proven to be an effective by previous work. For example, FS \cite{han2019efficient} uses the ancestors of vertices to optimize the backtracking process, and VE \cite{kim2021versatile} performs special processing on one-degree vertices in the search process. Both rely on the query graph. However, When edge number of query graph increases, the efficiency of these algorithms will drop drastically. Often times, there is no performance gain on a complete graph, and it could even introduce additional overhead. Therefore, it would be desirable if we could manipulate the query graph to ensure such algorithms to have stable and efficient pruning capabilities. 
The performance of pruning pretty much determine the overall performance of subgraph matching algorithms. 
Static pruning method will use check operation to remove infeasible candidate vertices. And in iterative method, previous work ignored the locality for check operation like GQL.

So it is important to reduce redundant $check$ or replace it by a fast pre-check method. 


\textit{Challenge 1: Can we transform the query graph without altering the original subgraph matching problem? } 

\textit{Our Solution: Dynamic Graph Editing. } With the given subgraph matching problem, it is difficult or at least not obvious that there is an answer to the question.
We formulate a problem: Assignment under Conditional Candidate Set(ACCS), which is proven to be equivalent to Subgraph matching problem, we propose Dynamic Graph Editing(DGE) that is for the first time designed to tailor the query graph. In general, DGE occurs in the case of partial matching results with several enumerated query vertices. And it will identify the removable edges in the query graph that follows the \textbf{Edge deleting Rule} explained in Section 6.3. 
The Rule is based on whether the Valid Candidate Set (VCS) will change after the intersection operation of the Conditional Candidate Set (CCS). The definition of VCS and CCS will be given in Section 6.

Finally, we propose the DGEE(DGE Enumeration) algorithm and carry it out at the $enumeration$ stage. Since the set-intersection based enumeration requires a step-by-step set intersection operation, it will not introduce the extra computation and maintain the  complexity of the original enumeration algorithm. Moreover, it exhibits excellent pruning capability. 

\textit{Challenge 2: Can we reduce or even eliminate redundant $check$  in iterative method? } 

\textit{Our Solution:  Locality Aware Exploration. } 
Principle of Locality is common foundational concepts used in computer system, like cache memory access \cite{li2015locality}.
It can be further characterized by two types of locality: temporal locality and spatial locality. 
And we take this concept into designing an efficient iterative method by Neighborhood Update and pre-check.
Algorithm like GQL\cite{he2008graphs} takes most time on $check$ operation: the $SemiMatching$. It runs in $O(|V(q)|^2*\max {d(v)})$ time. 
First, only the neighbors of updated vertex need to be considered in next iteration, which leads to less number of calls to \textit{SemiMatching}. Then, to preserve the last round matching result and do a pre-check. Since the time complexity of the pre-check is only $O(|V(q)|)$, it is much efficient than doing a full check.



Our contributions are the following:

\begin{enumerate}
    \item We reduce the subgraph matching problem to an ACCS (Assignment under Conditional Candidate Set) Problem. To our best knowledge, it is the first time the problem gets proposed. And based on ACCS, we give the Dynamic Graph Editing theorem and its proof. 
    \item We propose the Dynamic Graph Editing Enumeration algorithm in for subgraph matching. It exhibits stable performance under different static pruning methods, and outperforms the state-of-the-arts on real-world graphs. The baseline is the best performed out of eight (QuickSI, GraphQL, CFL, CECI, DAF, RI, VF2++, Glasgow) algorithms.
    \item We propose a fast GQL static pruning with pre-check enabled which is optimized by observing data locality effect. Our experiments show that it is the strongest static pruning strategy.
    \item We design an efficient subgraph matching procedure which demonstrates SOTA performance. 
\end{enumerate}

The rest of paper is organized as follows. Section 2 give notations and the problem statement. Section 3 is related works. Section 4 is the overview of our algorithm. Section 5 describes the static pruning and fGQL. Section 6 gives the definition of DGE and the enumeration algorithm DGEE. In Section 7, we compare the proposed algorithm with SOTA algorithms and present the performance analysis.  Section 8 concludes the paper.

\section{Preliminaries}
\subsection{Notations}
Table \ref{tab:notation} lists the notations used in this paper. 
In the article, unless otherwise specified, $u$ represents the vertex in the query graph, and $v$ represents the vertex in data graph. 

\begin{table}[h]
    \caption{Notations frequently used in this paper}
     \begin{tabular}{l |l}
     \hline \\
     Symbol & Definition \\ 
     \hline\hline
     $g, q, G$ & graph, query graph, data graph \\ 
     $u, v$ & $u, v$ is the vertex \\
     $V(g), E(g)$ & vertex set of graph g, edge set of graph g \\
     $d(v)$ & degree of vertex $v$ \\
     $C(u)$ & candidate set of vertex $u$ \\
     $C_{Algorithm}(u)$ & the candidate set of u in a specific Algorithm \\
     $C^{u_n,v_n}_u$ & candidate set of vertex $u$ while $un$ matches $vn$ \\
     $\Sigma$ & label set\\
     $L(u)$ & label of vertex u \\
     $\varphi$ & the set of enumerated vertices \\
     $f$ & mapping function f: $V(G_1) \rightarrow  V(G_2)$ \\
     $N(u)$ & neighbor vertices of vertex $u$ \\
    \hline
    \end{tabular}
    \label{tab:notation}
\end{table}

\subsection{Problem Statement}
This paper focuses on \textbf{Exact Non-Induced Subgraph Matching} problem, and the graph refers to undirected attributed graph. For the rest of the paper, we use subgraph matching and graph respectively. 

\textbf{Graph} Graph is represented as a triplet $G=(V, E, L)$, where $V$ is a set of vertices, $E$ is a set of edges, and $L$ is an attribute mapping function, which maps vertices to an attribute. For any edge $e \in E$, it consists of a pair of vertices $(v_i, v_j)$, where $v_i, V_j \in V$. Sometimes we will omit $L$ in $G$'s representation in writing as $G=(V,E)$.

\textbf{Subgraph} A graph $G'=(V', E')$ is a subgraph of graph $G=(V,E)$ that satisfies: 1) $V' \subseteq V$. 2) $E' \subseteq E$. And they share the same attribute mapping function.
 
\textbf{Exact Non-Induced Subgraph Matching} Given a data graph $G$ and a query graph $q$, exact subgraph matching requires that there is a bijective function $f$ between the query graph $q$ and the subgraph of data graph $G'$, where $G'=(V',E',L')$, $G' \subseteq G$, satisfying following constraints:

\begin{enumerate}
    \item f: $V(q) \rightarrow V(G')$.
    \item $\forall u \in V(q), L_q(u)= L_{G'}(v)$.
    \item $\forall (v_1, v_2) \in E(q) \rightarrow (f(v_1), f(v_2)) \in E(G')$.
\end{enumerate}

To further clarify, there could be multiple subgraph isomorphic embeddings mapping to the same subgraph, and they will be all considered as found solutions.

As shown in figure \ref{fig:SM-example}, the matching subgraph found in data graph for the query graph is $G'=\{V', E'\}$, where $V'=\{v_0, v_1, v_2\}$, $E'=\{\{v_0, v_1\}, \{v_0, v_2\}, \{v_1, v_2\}\}$. 

\begin{figure}[h]
  \centering
  \includegraphics[width=0.7\linewidth]{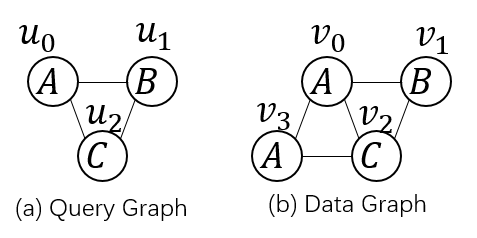}
  \caption{Example of subgraph matching problem with query graph and data graph}
  \Description{Example of subgraph matching}
  \label{fig:SM-example}
\end{figure}

The basic Depth-First-Search (DFS) process is elaborated in figure \ref{fig:DFS-example}. The whole process is a tree-based search. It start from $u_0$, and chooses a feasible matching vertex $v_0$. And it iterates until reaching the leaf vertices. The dotted lines are the branches that failed to match. After the search process completes, there is one solution found.

\begin{figure}[h]
  \centering
  \includegraphics[width=\linewidth]{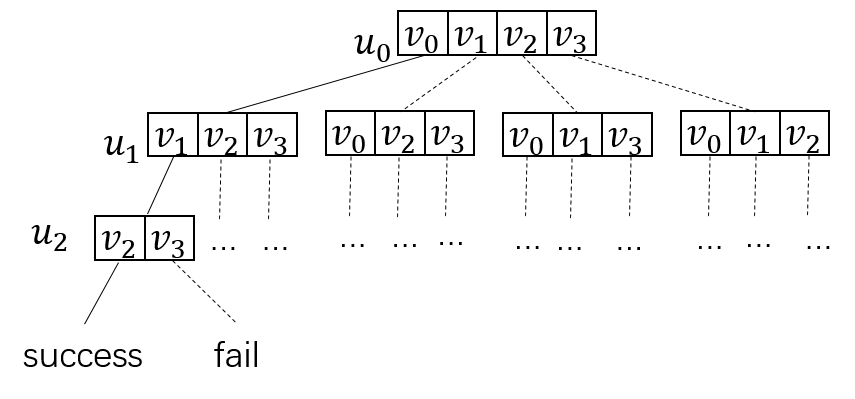}
  \caption{Example of search process in subgraph matching}
  \Description{Example of search process}
  \label{fig:DFS-example}
\end{figure}

\section{Related Work}
\textbf{Subgraph Matching} 
Subgraph matching is an NP complete problem. The Ullmann \cite{ullmann1976algorithm} algorithm uses a depth-first search strategy to enumerate subgraphs that match the vertices of the query graph in data graph. The \textit{enumeration based} algorithm usually consists of three phases: preproccessing, reordering and DFS (depth-first search). The DFS is also called enumeration process in this paper. 

Pruning is the most crucial optimization technique for enumeration-based method, which can be roughly categorized into static pruning and dynamic pruning. 

\textbf{Static Pruning}
It is unexceptional that constraints of neighboring vertices are considered as the pruning filter by existing subgraph matching algorithms. 
Versatile Equivalences \cite{kim2021versatile} uses neighbor-safety rules in pruning. GraphQL \cite{he2008graphs} extracts the neighbors around the vertices of the query graph and data graph, and constructs a bipartite graph to judge whether it has a semi-matching relationship.
And there are many pruning algorithms based on spanning tree $q_T$,like Turbo\textsubscript{iso} \cite{han2013turboiso}. 
With the increase of computer memory capacity, some algorithms using auxiliary arrays with the expense of larger space complexity have also been proposed. CECI \cite{bhattarai2019ceci} (Compact Embedding Cluster Index) designs an auxiliary data structure $\mathscr{A}$ with $O(|E(q)|*|E(G)|)$ space complexity. CFL \cite{bi2016efficient} designs a tree-structured auxiliary data structure $\mathscr{A}$ with $O(|V(q)|*|E(G)|)$ space complexity. Both methods are built upon  spanning trees. The spanning tree-based method requires at least one round of forward and backward updates. 

\textbf{Dynamic Pruning}
Based on Ullmann's algorithm, VF2 and VF3 \cite{cordella2004sub, carletti2017introducing} introduce the deletion of candidate vertices that are not adjacent to the matched vertices. Han \cite{han2019efficient} proposes a DAF algorithm with FS (failing set) technique that can be integrated in many \textit{enumeration based} algorithms. However, when the query graph is a dense graph, FS produces little effect. Furthermore, FS could be noneffective when the query graph is a complete graph. Versatile Equivalences\cite{} seeks dynamic equivalence to avoid searching for equivalent vertices. Recently, the mainstream enumeration method is set-intersection based, and this method will build an auxiliary data structures maintaining edge information, like Turbo\textsubscript{iso},DAF and CFL.

\textbf{Set Intersection} Set intersection plays an important role in subgraph matching algorithm designing. And it is also widely used in other graph algorithms such as triangle count \cite{mowlaei2017triangle}, maximum clique count \cite{blanuvsa2020manycore}. QFilter \cite{han2018speeding} is proposed to speed up set intersection with BSR structure and SIMD instruction. In this paper, we choose Qfilter to carry out the intersection operation.






\section{Algorithm Overview}
In this paper, we propose a new strategy for solving subgraph matching problem with introduction of ACCS(Assignment under Conditional Candidate Set) problem. 

First, CS (Candidate Set) and CCS (Conditional Candidate Set) can be obtained through polynomial time during preprocessing. 
Before solving the problem directly, we reduce the problem and turn it into a ACCS Problem and assign values to the vertices of query graph through CS and CCS. To solve ACCS,  data graph is no longer needed after construction of CS and CCS.
Dynamic Graph Editing (DGE) is proposed to  manipulate the query graph to speedup searching process. We present DGEE for solving ACCS, which is a new enumeration-based algorithm that integrates DGE and FS. And we will prove that we obtain the same results of the subgraph matching problem with given query graph $q$. 

The proposed subgraph matching algorithm consists three parts: Preprocessing, Reordering and Enumerating as shown in algorithm \ref{alg:pipeline}. These three modules are executed sequentially.

\begin{algorithm}
\SetAlgoLined
\KwResult{All matches }
 $Preprocessing(q, G)$\;
 $Reordering(q)$\;
 $Enumerating(q, G)$\;
 \Return matches\;
 \caption{Main algorithm($q, G$)}
 \label{alg:pipeline}
\end{algorithm}

\textbf{Preprocessing.} In this module, we propose a fast algorithm fGQL. Candidate set $C(u)$ of each vertex $u$ in query graph $q$ and conditional candidate set $C^{u_n, v_n}_u$ will be calculate. And this part is also called static pruning, since the candidate vertices removed in this part do not in the enumeration process. The details will be discussed in Static Pruning Section. 

\textbf{Reordering.} In this module, we use RI strategy. It is a part of dynamic pruning, and it can affect the depth of the pruning that occurs in search process. The paper \cite{sun2020memory} shows RI still an efficient reordering method. First, it selects maximum degree vertex as root $u_0$. Then it iteratively select $u_k = \arg\max_{u \in N(\varphi)-\varphi}{|N(u) \bigcup{\varphi}|}$ at level $k$, where $\varphi = \{u_0, u_1, ..., u_{k-1}\}$. And it breaks ties by: 

1) $u_k = \arg\max_{u \in N(\varphi)-\varphi}{\{|N(N(u) \bigcap{V(q)-\varphi}) \bigcap{\varphi}|\}}$. 

2) $u_k = \arg\max_{u \in N(\varphi)-\varphi}{\{|(N(u)-\varphi)\bigcap{N(\varphi)}|\}}$.

\textbf{Enumerating.} In this module, we propose DGEE. It is the main part of the algorithm. It will maintain a valid candidate set for each level. And the valid candidate set is defined as: $VC(u)=\bigcap_{u_n\in N(u)}{C^{u_n,v_n}_u}$. Since the main operation is set intersection, it can be accelerated by QFilter \cite{han2018speeding} with BSR (Base and State Representation) data structure and SIMD instructions. In addition, at each level, the algorithm will dynamically edit the query graph by deleting and adding edge to decrease search space, while ensuring that all solutions can be found.

\section{Static Pruning}
In this section, we propose an efficient static pruning algorithm based on GQL \cite{he2008graphs}, called fast GQL (fGQL). The algorithm adapts two check rules to remove the infeasible candidate vertices and two locality strategies to reduce redundant checks.



\subsection{Rule 1: Neighbor Label Frequency}
Neighbor Label Frequency (NLF) is used to remove infeasible candidate vertices from the full set $V(G)$ in a coarse-grained manner.

It will count the label frequency of the neighbors of each vertex $u$ and the label frequency of the neighbors of the candidate vertex $v$ in $V(G)$, so that each label frequency of the latter is no less than the former.

The candidate set can be formulated as: $C_{NLF}(u)=\{v | v\in V(G), L(v)=L(u), |N(u,l)| \leq |N(v,l)|, \forall l \in \Sigma \}$. The $N(u,l)$ can be built in reading graph stage with $O(|V(G)|* |\Sigma|)$ space complexity. 
And the time complexity of calculating candidate set is $O(|V(q)|*|V(G)|*|\Sigma|)$. 
In terms of implementation, the set of corresponding vertex for each label in the data graph can be reprocessed. In this way, if the number of each label is uniform, the time complexity can be reduced to $O(|V(q)|*|V(G)|)$.

\subsection{Rule 2: Neighborhood Matching}
Neighborhood Feasible is used to remove infeasible vertices from the complete set $V(G)$ in a fine-grained manner.

It will count the neighbors of each vertex $u$ and the neighbors of the candidate vertex $v$ in $u$' candidate set and get a bipartite graph. The left part is $N(u)$, and the right part is $N(v)$. For vertex $u_n \in N(u)$ and $v_n \in N(v)$, there is an edge $(u_n. v_n)$ if $v_n \in C_{fGQL}(u)$. And the set $C_{fGQL}(u)$ is initialized by $C_{NLF}(u)$.
$SemiMatching$ is used to find a matching relationship between the two neighbor sets can be satisfied. 



The candidate set can be formulated as: $C_{fGQL}(u) = \{ v| v \in C_{NLF}(u)$,  $ SemiMatching (N(u), N(v),$ $E=\{(u_n, v_n)| v_n \in C_{fGQL}(u_n)\}) \}$ after converged, where $SemiMatching(V1,V2,E)$ is used to judge whether the bipartite graph matches. It is an iterative algorithm, and the expected time complexity is $O(k*|C_{NLF}|*\max d(u)*\max d(v))$, where $k$ is iterate number, $v \in V(q)$ and $u \in V(G)$ . In real-graph testing, the $k$ usually is small, and can be regarded as a constant. 



\subsection{fast GQL}
\textit{SemiMatching} will occupy almost all computational load. And it takes $O(|V(q)|^2*\max {d(v)})$ time complexity. Higher performance can be achieved by reducing or accelerating \textit{SemiMatching} calculation.

Through the principle of locality, there are two types redundant operations that can be avoided in GQL after an iteration. First, for spatial locality of the removed candidate vertices, if the candidate sets of $u$ and its neighbors $N(u)$ are not updated in last iteration, there will be a repetitive and ineffective operation for these vertices. 
Second, for temporal locality of the maintained candidate of vertex $u$, the $SemiMatching$'s result can be same between two iterations.

So two techniques are used to solve two redundant operations: \textbf{Neighborhood Update} and \textbf{pre-check}. 
First, only the neighbor vertices of updated vertex need to be considered in next iteration. It can reduce the number of calls to \textit{SemiMatching}. Second, store last matching information and do a pre-check before the next matching. And it can accelerate \textit{SemiMatching} process.

The proposed fast algorithm is called fGQL in algorithm \ref{algo:fGQL}. The first line of the code is the result of NLF. 
Line 2 initializes the queue for all vertices. 
Lines 3-6 have similar functions to lines 2-3 in GQL, enumerating $(u,v)$ pair. 
Lines 7-9 use the last matching result as a pre-check, which is to avoid the second redundancy situation mentioned above. 
The process of saving information occurs in Line 16. 
Lines 12-14 add $u$ and neighbors $N(u)$ to the queue for subsequent iteration after $C(u)$ updated.
This is to tackle the first redundancy operation. 
Meanwhile, Line 12 is used to avoid repeated update. Since more than one $v$ may be removed from $C(u)$ in one iteration, but the update only needs to be done once.

In \textit{pre-check}, given a $(u,v)$ pair, last matching result of the left part in bipartite graph will be stored for looking up. The number of pairs is $O(|V(q)|* max C(u))$. The number of vertices of the left part is $O(|V(q)|)$. So the space is $O(|V(q)|*\max {C_{fGQL}(u)} * |V(q)|)$. By looking up the last matching result for each vertex in the left part, the time complexity can be $O(|V(q)|)$. 
After passing the \textit{pre-check}, it can save a lot time for $Semi-Match$ test.


\begin{algorithm}
\SetAlgoLined
\KwResult{The candidate set after fGQL}
 $C_{fGQL} = C_{NLF}$\;
 queue = $V(q)$  \;
 \While {queue is not $\emptyset$}{
    u = queue.front()\;
    queue.pop()\;

     \For{$v$ in $C_{fGQL}(u)$} {
        \If{pre-check last matches}{
            \textbf{Continue}\;
        }
        \eIf{not $SemiMatching(N(u), N(v), E=\{(u_n, v_n)| v_n \in C_{fGQL}(u_n)\})$}{
            $C_{fGQL}(u) -= \{v\}$\;
            \If{$u$ is unupdated in this iterate}{
                add $\{N(u), u\}$ in queue\;
            }
        }{
            store matching information\;
        }
     }
 }
 \caption{fGQL(q,G)}
 \label{algo:fGQL}
\end{algorithm}

In addition, in theory, we analyzed and compared it with other popular static pruning algorithms.

\textbf{CFL.} The candidate set is defined as: $C_{CFL}(u) = \{ v |v \in C_{NLF}(u) $ $\bigcap_{u_n \in N(u)}{N(C_{CFL}(u_n))}, C_{CFL}(u_n) \cap{N(v)} \neq \emptyset  \} $ after converged. It is an iterative algorithm and the expected time complexity is $O(k * |E(q)| * |E(G)||)$, where $k$ is iterate number. The default value is 1.  

\begin{theorem}
$C_{GQL}(u) \subseteq C_{CFL}(u).$
\end{theorem}

\begin{proof}
These two methods are iterative algorithms, and their starting set is $C_{NLF}(u)$. And we only need to prove that each element in $C_{GQL}(u)$ is still in $C_{CFL}(u)$. 

For $\forall v \in C_{GQL}(u),$ $v_n \in N(v), u_n \in N(u)$. $v_n$ is an arbitrary neighbor of $v$, and $u_n$ is an arbitrary neighbor of $u$. There is a mapping $f$ from vertex set \{u, N(u)\} to set \{v, N(v)\}.

Now we need to prove that 

1) $v \in \bigcap_{u_n \in N(u)}{N(C_{CFL}(u_n))}$ for $\forall u_n \in N(u)$, 

2) $C_{CFL}(u_n)$ $\cap{N(v)} \neq \emptyset$, for $\forall u_n \in N(u)$.

We have $f(u_n) \in C_{NLF}(u_n)$ , $f(u_n) \in N(v)$, and $v \in N(f(u_n))$. 

Then, it can be written as $f(u_n) \in C_{NLF}(u_n)\cap N(v) $.

So, $C_{NLF}(u_n)\cap N(v) \neq \emptyset$. The Part 2 is proven.

Since $f(u_n) \in C_{NLF}(u_n)$ and $v \in N(f(u_n))$, we have $v \in N(C_{NLF}(u_n))$. 

Therefore, $v \in C_{CFL}(u)$. The Part 1 is proven.

Hence, $C_{GQL}(u) \subseteq C_{CFL}(u)$.
\end{proof}

In fact, the collection obtained by fGQL is still a subset of latest static pruning algorithm like CECI and DAF. So, for now, it is the strongest pruning technique in the most practical algorithm.


\section{Dynamic Graph Editing in Dynamic Pruning}
In this section, we first introduce set-intersection-based enumeration.
And then, on this basis, the subgraph matching problem can be reduced to the ACCS Problem. Then put forward the DGE technique in ACCS Problem. Finally, a new enumeration algorithm DGEE is given. 

The dynamic pruning includes reordering and enumeration process. However, the reordering is not the focus of our research. Therefore, RI is used as reordering method in our algorithm.

\subsection{Set-Intersection based enumeration}
In this subsection, we will introduce the basic workflow and two important sets.

Most of the advanced enumerated-based subgraph matching algorithms can be classified as Set-Intersection based method. It can be described as algorithm \ref{algo:set-intersection-based}. Lines 1-4 indicate that the matching subgraph has been found. Line 5 is to update the Valid Candidate Set. And Lines 6-9 are to enumerate candidate vertex in data graph and enter the next layer of the search.

\begin{algorithm}
\SetAlgoLined
\KwResult{The matching results}
 \If{$k=|V(Q)|$ } {
    Report it\;
    \Return\;
 }
 ${VC}_{u_k} = \bigcap_{u_n \in N(u_k) \cap{\varphi}}{C^{u_n,v_n}_{u_k}}$\;
 \For{$v_k \in {VC}_{u_k} $ } {
    $M' = M \cup \{(u_k, v_k)\}$\;
    Enumerate(q, G, k+1)\;
 }
 \caption{Enumerate(q,G,k) }
 \label{algo:set-intersection-based}
\end{algorithm}

In enumeration process, there are two important sets Valid Candidate Set and Conditional Candidate Set. Although it is a set, it is implemented with an array. 

\textbf{Valid Candidate Set. } This set can satisfy that the edge constraint among the enumerated vertices $\varphi$. For each vertex, it is calculated by the neighbor's CCS (Conditional Candidate Set). During the enumeration process, only the neighbors that have been enumerated are needed. It is calculated as: $VC(u) = \bigcap_{u_n\in (N(u)\cap \varphi)}{C^{(u_n,v_n)}_u}$. And $VC(u_0)=C(u_0)$, where $u_0$ is the start vertex in searching.

\textbf{Conditional Candidate Set. } It represents the candidate set of vertex $u$ under the condition of a given matching relationship of $(u_n, v_n)$, where $u_n$ is the $v$'s neighbor and $v_n$ is the corresponding matched vertex of $u_n$. And it is defined as: $C^{u_n,v_n}_u = C(u) \cap {N(v_n)}$. Naturally, it satisfies the edge constraint. The most important thing here is to do the intersection operation with $N(v_n)$, which is an critical condition for the establishment of the DGE lemma. 

In enumeration process, the algorithm in this paper takes RI's query plan and then starts a depth-first search from the start vertex $u_0$. In each level of the search process, the algorithm will 1) update the current VC, 2) select a candidate vertex from the VC, and then 3) enter the next layer. When the VC at a level is empty, it means that the branch can't find a solution, and it backtracks to its previous state. 

\subsection{ACCS Problem}
In this subsection, we will reduce the subgraph matching problem into ACCS Problem. It is a new problem we proposed. DGE is based on this reduced problem. 

After preprocessing, the Conditional Candidate Set will be stored in an array, and our algorithm only needs the information of the query graph, not the data graph, based on algorithm \ref{algo:set-intersection-based}. So the problem can be reduced to: 

\begin{problem}[Assignment under Conditional Candidate Set]
    Given candidate set, conditional candidate set and query graph $q$, for each vertex $u_i$ in graph $q$, assign a value $v_i$ to $u_i$ from $C(u)$, which satisfies 1) $ v_i \in \bigcap_{u_n \in N(u_i)}{C^{u_n, v_n}_{u_i}}$,  and 2) $v_i \neq v_j$, $\forall 0 \leq j \leq |V(q)|-1$, $j \neq i$. 
    \label{problem: NewProblem}
\end{problem}

Below we give the proof that the solution of Problem \ref{problem: NewProblem} is same as subgraph matching problem's.

\begin{proof}
Let subgraph matching problem's solutions be $S$, and $s=\{(u_i,v_i)$, $i=0,...,|V(q)|-1\}$ is one solution. ACCS Problem's solutions be $P$, and $p=\{(u_i, v_i), i=0,...,|V(q)|-1\}$ is one solution. We use $s_i$ to represent $v_i$ in $s$, and $p_i$ to $v_i$ in $p$, since their $u_i$ is equal.

There are three constraints the in ACCS Problem: 1)  $p_i \in C(u_i)$, 2) $p_i \in \bigcap_{u_n \in N(u_i)}{C^{u_n, p_n}_{u_i}}$. 3) $p_i \neq p_j$, $\forall 0 \leq i,j \leq |V(q)|-1$, $i \neq j$.

The subgraph matching problem's constraints are: 1) $s_i \in C(u_i)$ 2) $s_i \in N(s_j)$, $\forall u_j \in N(u_i)$. The former contains matching constraints on vertex, and the latter contains matching constraints on edges. 3) $s_i \neq s_j \forall 0 \leq i,j \leq |V(q)|-1$, $i \neq j$.

Their first and third constraints are same, and we only need to prove that the second constraint can be deduced from each other combined with known conditions. 

a) Now, we begin to prove that $p$ is one solution in $S$, i.e.,  $p_i \in \bigcap_{u_n \in N(u_i)}{C^{u_n, v_n}_{u_i}}$ $\Longrightarrow$ $p_i \in N(p_j)$, $\forall u_j \in N(u_i)$.

Because of Equation (1), and $p_i \in \bigcap_{u_n \in N(u_i)}{C^{u_n, p_n}_{u_i}}$.

\begin{equation}
\begin{aligned}
     \bigcap_{u_n \in N(u_i)} C^{u_n,p_n}_{u_i} &= \bigcap_{u_n \in N(u_i)} C(u_i) \cap {N(p_n)} \\ 
    &=  C(u_i) \bigcap_{u_n \in N(u_i)} {N(p_n)}  \\
    &\subseteq \bigcap_{u_n \in N(u_i)} {N(p_n)} \\
\end{aligned}
\end{equation}

We have $p_i \in N(v_n)$ for $\forall u_n \in N(u_i)$. So this part is proven.

b) Now, we begin to prove that $s$ is one solution in $P$, i.e.,  $s_i \in N(s_j)$, $\forall u_j \in N(u_i)$ $\Longrightarrow$ $s_i \in \bigcap_{u_n \in N(u_i)}{C^{u_n, s_n}_{u_i}}$.

From $s_i \in N(s_n)$, for $\forall u_n \in N(u_i)$.  we have 

\begin{equation}
s_i \in \bigcap_{u_n \in N(u_i)} {N(v_n)}
\notag
\end{equation}

Because of $s_i \in C(u_i)$ and the equation above, we can get that 

\begin{equation}
s_i \in C(u_i) \bigcap_{u_n \in N(u_i)} {N(v_n)}
\notag
\end{equation}

Therefore, 

\begin{equation}
s_i \in \bigcap_{u_n \in N(u_i)}{C^{u_n, s_n}_{u_i}}
\notag
\end{equation}


So this part is proven.

Thus, the second constraint can be deduced from each other. And the solution of Problem \ref{problem: NewProblem} is same as subgraph matching problem's solution.

\end{proof}

The difference between the ACCS Problem and the subgraph matching problem is that the information of the data graph is replaced with CCS, which weakens the strict matching constraints of edges, making edge deleting technique feasible. Moreover, in addition to neighbor information, CCS can also introduce more pruning techniques in future research based on its definition. 

\subsection{Dynamic Graph Editing}
In this subsection, we will introduce several basic concepts of Dynamic Graph Editing like Edited Graph, Edge adding, Edge deleting. And then, we propose the Dynamic Graph Editing theorem. To demonstrate the workflow, we use an example. 

In subgraph matching, the role of edges is to transport vertex-to-vertex constraints. Intuitively, the fewer edges in a graph, the fewer constraints, and the easier the matches will be find. 

More generally, if an edge cannot transport a vertex-to-vertex constraints, the edge can be deleted. On the other hand, if there is vertex-to-vertex constraints without an edge connected, we can insert an edge to the graph. 

DGE is proposed to reduce edges. But due to the another constraint of subgraph matching that the vertex in the data graph can only be mapped by one vertex in query graph, some edges must be added in this scope in order to satisfy the constraint. 

\begin{figure*}[h]
  \centering
  \includegraphics[width=\linewidth]{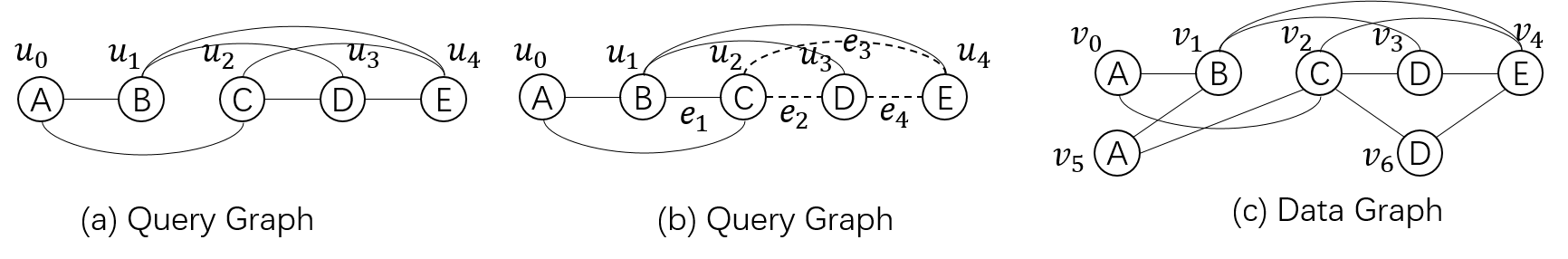}
  \caption{Example of DGE. (a): the original query graph; (b): an edited graph in searching processing; (c): the data graph.}
  \Description{Example of DGE}
  \label{fig:DGE-example}
\end{figure*}

\begin{algorithm}
\SetAlgoLined
\KwResult{Update edited graph and VC}
 $E(eg_k) = E(eg_{k-1})$\;
 \For( //calculate $E_{del}$) {$u_n \in N(u) - \varphi_k$} {
    $VC(u_n) = C^{u,v}_{u_n} \bigcap VC(u_n)$\;
    \If{size of $VC(u_n)$ changed} {
        $E(eg_k) -= {(u,u_n)}$\;
    }
 }
 \For(//calculate $E_{add}$){$u_n \in V(q) - \varphi_k$}{
    \If{$v_k \in VC(u_n)$}{
        $VC(u_n) -= {v_k}$\;
        $E(eg_k) += {(u,u_n)}$ \; 
    }
 }

 \caption{DGE($eg_k, {VC}_u, u, \varphi_k$)}
 \label{algo:DGE}
\end{algorithm}

In DGE, there are three concepts: Edited Graph, Edge adding Rule and Edge deleting Rule. 

\textbf{Edited Graph. } We call it Edited Graph after editing edges in query graph. After enumerating a new vertex $u$, a candidate vertex $v$ will be determined. Based on this new knowledge, we can review the edges information and edit graph to reduce the size of the subsequent search space. This process is called DGE (Dynamic Graph Editing). We use $eg_k$ to represent the Edited Graph at level $k$ in DFS (depth-first search). $V(eg_k) = V(eg_{k-1})$ and $E(eg_k) = E(eg_{k-1}) + E_{add} - E_{del}$. $E_{add}$ and $E_{del}$ are calculated by Edge adding and Egde deleting respectively. And they are two cores of DGE technique. 

\textbf{Edge adding Rule. } The adding edge set is defined as: $E_{add}$ = $\{ (u_i$, $u_j) |$ $u_i \in \varphi$,$u_j \in V(q)-\varphi$,  $\exists v_i \in VC(u_j) \}$. Here $u_i$ is the matched vertex, the mapping is $v_i$. And $u_j$ is the unmatched vertex. $VC(u_j)$ needs special explanation, it is not a constant. It will be mentioned together after explaining Edge deleting. The adding edge set's definition means that used vertex in data graph cannot be mapped by other vertices. In fact, such edges have been implicitly added to the graph in previous enumeration algorithms. Meanwhile, this is also the lowest constraints of $E_{add}$. Such an operation is possible to turn the graph into a complete graph, but it is unavoidable. 

\textbf{Edge deleting Rule.} The deleting edge set is defined as: $E_{del}$ = $\{(u_i$, $u_j) $ $| u_i \in \varphi, VC(u_j) \cap{C^{u_i,v_i}_{u_j}} = VC(u_j)$, $u_j \in V(q) - \varphi \}$.  Here $u_i$ is the matched vertex, the mapping is $v_i$. And $u_j$ is the unmatched vertex. It needs to meet the condition $VC(u_j) \cap{C^{u_i,v_i}_{u_j}} = VC(u_j)$.
It means that if an edge does not affect the set of feasible candidate vertices of another vertex, then this edge can be safely removed from the graph. 

There is a problem with the value of $VC(u_j)$ in Edge adding and Edge deleting. 
By definition, when the mapping of $u \in \varphi$ is determined, $VC(u_j)$ can be calculated as: $VC(u_j) = \bigcap_{u_i \in \varphi \cap{N(u_j)}}{C^{u_i,v_i}_{u_j}}$. But under this definition, $E_{del}$ will be an empty set. In fact, the order of intersection  process will effect the $E_{add}$ and $E_{del}$. 
Since $E_{add}$ and $E_{del}$ are iterative, so we let 

$VC(u_j) = \bigcap_{u_i \in (\varphi \cap{N(u_j)} - \{u_{|\varphi|-1}\})}{C^{u_i,v_i}_{u_j}}$. The index of vertices starts from 0, and the first vertex is $u_0$. Here ,the last matching vertex $u_{|\varphi|-1}$ in the intersection operation is removed. 

In this case, $E_{add} = \{ (u_i$, $u_j) |$ $u_i =u_{|\varphi|-1} $,$u_j \in V(q)-\varphi$,  $\exists v_i \in VC(u_j) \}$,  $E_{del}$ = $\{(u_i$, $u_j) $ $| u_i=u_{|\varphi|-1}, VC(u_j) \cap{C^{u_i,v_i}_{u_j}} = VC(u_j)$, $u_j \in V(q) - \varphi \}$. 

Based on the ACCS Problem, we propose the lemma Dynamic Graph Editing as follows.

\begin{theorem} [Dynamic Graph Editing]
In ACCS Problem, given the assigned vertices set $\varphi_k=\{u_0,u_1,..,u_{k-1}\}$, the solutions of $g$ is equal to $eg_{k}$'s solutions. 
\label{lemma:edit-graph}
\end{theorem}

\begin{proof}
Since $eg_0=g$, we only need to prove that $eg_k$'s solutions is equal to $eg_{k-1}$'s solutions. And it can be proved by induction that the solutions of $eg_k$ and $q$ are the same.

Let $eg_{k-1}$'s solutions be $S$, and $s=\{(u_i,s_i)| i=0,...,|V(eg_{k-1})|-1\} $ is one solution, $eg_k$'s solutions be $S'$, and $s'=\{(u_i, s'_i)| i=0,...,|V(eg_k)|-1\}$ is one solution, where $V(eg_{k-1})=V(eg_k)=V(q)$.

There are three constraints of $eg_{k-1}$'s solution: 1)  $s_i \in C(u_i)$, 2) $s_i \in \bigcap_{u_n \in N_{k-1}(u_i)}{C^{u_n, s_n}_{u_i}}$. 3) $s_i \neq s_j \forall 0 \leq i,j \leq |V(q)|-1$. 3) $s_i \neq s_j \forall 0 \leq i,j \leq |V(q)|-1$.

There are three constraints of $eg_k$'s solution: 1)  $s'_i \in C(u_i)$, 2) $s'_i \in \bigcap_{u_n \in N_k(u_i)}{C^{u_n, s'_n}_{u_i}}$. 3) $s'_i \neq v_j \forall 0 \leq i,j \leq |V(q)|-1$. 3) $s'_i \neq s'_j \forall 0 \leq i,j \leq |V(q)|-1$.

And $N_k = N_{k-1} + V_{add} - V_{del}$. Since the edges here have the same vertex $u_k$ in $E_{add}$ and $E_{del}$, for simplicity, the set of another vertex in each edge in $E_{add}$ is called $V_{add}$, e.g., $E_{add}$ is $\{(u_k,v_1), (u_k,v_2)\}$, and $V_{add}$ is $\{v_1, v_2\}$.

Their first and third constraints are same, and we only need to prove that the second constraint can be deduced from each other combined with known conditions.

a) Now, we begin to prove that $s$ is one solution in $S'$, i.e.,  $s \in S$, and $s_i \in \bigcap_{u_n \in N_{k-1}(u_i)}{C^{u_n, v_n}_{u_i}} \Longrightarrow s_i \in \bigcap_{u_n \in N_k(u_i)}{C^{u_n, v_n}_{u_i}}$.

The above conclusion cannot be proven directly, and the third condition is needed to prove it.

By utilizing third condition, We can rewrite the second condition as 

\begin{equation}
s_i \in \bigcap_{u_n \in N_{k-1}(u_i)+V_{add}}{C^{u_n, v_n}_{u_i}}
\notag
\end{equation}

Because of following equation

\begin{align*}
    \bigcap_{u_n \in N_{k-1} + V_{add}}{C^{u_n, v_n}_{u_i}} &= \bigcap_{u_n \in N_{k-1} + V_{add} - V_{del}}{C^{u_n, v_n}_{u_i}}
\bigcap_{u_n \in E_{del}}{C^{u_n, v_n}_{u_i}}\\
&=\bigcap_{u_n \in N_{k-1} + V_{add} - V_{del}}{C^{u_n, v_n}_{u_i}}
\end{align*}

Then, 
\begin{equation}
s_i \in \bigcap_{u_n \in N_{k-1} + V_{add} - V_{del}}{C^{u_n, v_n}_{u_i}}=\bigcap_{u_n \in N_k(u_i)}{C^{u_n, v_n}_{u_i}}
\notag
\end{equation}

And this part has been proven.


b) Now, we begin to prove that $s'$ is one solution in $S'$, i.e., $s' \in S'$, and $s'_i \in \bigcap_{u_n \in N_k(u_i)}{C^{u_n, v_n}_{u_i}} \Longrightarrow s'_i  \in \bigcap_{u_n \in N_{k-1}(u_i)}{C^{u_n, v_n}_{u_i}}$. 

Because of following equation

\begin{align*}
    \bigcap_{u_n \in N_k}{C^{u_n, v_n}_{u_i}} &= \bigcap_{u_n \in N_{k-1} + V_{add} - V_{del}}{C^{u_n, v_n}_{u_i}}\\
&=\bigcap_{u_n \in N_{k-1} + V_{add}}{C^{u_n, v_n}_{u_i}}\\
&\subseteq \bigcap_{u_n \in N_{k-1}}{C^{u_n, v_n}_{u_i}}
\end{align*}

Hence, 

\begin{equation}
s'_i \in \bigcap_{u_n \in N_{k-1}}{C^{u_n, v_n}_{u_i}}
\notag
\end{equation}

And this part has been proven.

Thus, the solution of $eg_k$ is equal to $eg_{k-1}$'s solution. 

\end{proof}

Because of Dynamic Graph Editing Lemma, we can replace query graph $q$ with edited graph $eg_k$ in subsequent search. 

The DGE process can be designed in Algorithm \ref{algo:DGE}. Let $eg_0$ = $q$. The vertices of edited graph is always $V(q)$, so we don't need to edit the vertices in the graph. Lines 2-7 calculate the edges that can be deleted and update the VC. Lines 8-13 calculate the edges that must be added and update the VC. 

The core operation of DGE is to edit the edges in the graph. It will not have a pruning effect in classical set-intersection-based algorithm, but it will have an effect combined with the pruning algorithm like FS, as we will introduce in the following subsection.

\textbf{Example of DGE. } Here, we present an example of DGE process. We have a Query Graph in figure \ref{fig:DGE-example}(a), and figure \ref{fig:DGE-example}(c) is a part of a Data Graph. Figure \ref{fig:DGE-example}(b) is the edited graph after a branch search. Let the search order be $\{u_0,u_1,u_2,u_3,u_4\}$, the candidate set be $C(u_0)=\{v_0, v_5, ...\}$, 
$C(u_1)=\{v_1,v_2, ...\}$, 
$C(u_2)=\{v_1, v_2, ...\}$, 
$C(u_3)=\{v_3, v_6, ...\}$, 
$C(u_4)=\{v_4, ...\}$ and the conditional candidate set be 
$C^{u_0, v_0}_{u_1}=\{v_1,v_2\}$,
$C^{u_0, v_5}_{u_1}=\{v_1, v_2\}$,
$C^{u_0, v_0}_{u_2}=\{v_1,v_2\}$,
$C^{u_0, v_5}_{u_2}=\{v_1,v_2\}$,
$C^{u_1, v_1}_{u_3}=\{v_3\}$,
$C^{u_1, v_2}_{u_3}=\{v_3,v_6\}$, 
$C^{u_1, v_1}_{u_4}=\{v_4\}$, 
$C^{u_1, v_2}_{u_4}=\{v_4\}$, 
$C^{u_2, v_1}_{u_3}=\{v_3\}$,
$C^{u_2, v_2}_{u_3}=\{v_3,v_6\}$, 
$C^{u_2, v_1}_{u_4}=\{v_4\}$, 
$C^{u_2, v_2}_{u_4}=\{v_4\}$, 
$C^{u_3, v_3}_{u_4}=\{v_4\}$, 
$C^{u_3, v_6}_{u_4}=\{v_4\}$. Assume that these values are obtained by preprocessing.

The enumeration steps are shown in table \ref{tab:DGE-example}. 
At Level 0 of search tree, the VC of each vertex is the same as candidate set. And at Level 1 it is determined that $u_0$ mapping to $v_0$, and the VC corresponding to neighbor vertices $u_1$ and $u_2$ will change. At Level 2 $u_1$ is determined to map to $v_1$, and the VC corresponding to neighbor vertices $u_3$ and $u_4$ will change. At this time, since $v_1$ is in $VC(u_2)$, an edge $(u_1,u_2)$ will be added in the graph. At Level 3, $u_2$ is determined to map to $v_2$. 
Since $VC(u_3)$ and $VC(u_4)$ have not changed, that is, the mapping of $u_2$ don't affect the vertices $u_3$ and $u_4$, and the edges $(u_2,u_3)$ and $(u_2,u_4)$ can be safely removed. At Level 4, $u_3$ is determined to map to $v_3$. Since $VC(u_4)$ have not changed, that is, the mapping of $u_3$ don't affect the vertices $u_4$, and the edges $(u_3,u_4)$ can be safely removed. Finally, at last level, $u_4$ is determined to map to $v_4$. At this point, the remaining edited graph is shown in figure \ref{fig:DGE-example}(b). $e_1$ is added in it, and $\{e_2$, $e_3$, $e_4\}$ is deleted compared with query graph.

\begin{table}[h]
  \begin{center}
    \caption{DGE-example in enumeration}
    \label{tab:DGE-example}
    \begin{tabular}{p{0.10\linewidth}|p{0.10\linewidth}|p{0.10\linewidth}|p{0.10\linewidth}|p{0.10\linewidth}|p{0.10\linewidth}|p{0.10\linewidth}}
      \hline
      Level &0 & 1 & 2& 3 & 4 & 5\\
      \hline
      \hline
        & - &  $(u_0, v_0)$ & $(u_1, v_1)$ & $(u_2, v_2)$ & $(u_3, v_3)$ & $(u_4, v_4)$\\
      \hline 
      $VC(u_0)$ & $\{v_0, v_5$, $...\}$ & - & - &  - & - & - \\
      $VC(u_1)$ & $\{v_1, v_2$,  $...\}$ & $\{v_1, v_2\}$ & - & - & - & - \\
      $VC(u_2)$ & $\{v_1, v_2$,  $...\}$ & $\{v_1, v_2\}$ & $\{v_2\}$ & - & - & - \\
      $VC(u_3)$ & $\{v_3, v_6$,  $...\}$ & $\{v_3, v_6$,  $...\}$ & $\{v_3\}$ & $\{v_3\}$ & - & - \\
      $VC(u_4)$ & $\{v_4$,  $...\}$ & $\{v_4$,  $...\}$ & $\{v_4\}$ & $\{v_4\}$ & $\{v_4\}$ & - \\
      \hline
      \hline
      $E_{add}$ & - & $\emptyset$ & $\{e_1\}$ & $\emptyset$ & $\emptyset$& $\emptyset$\\
      $E_{del}$ & - & $\emptyset$ & $\emptyset$ & $\{e_2,e_3\}$ & $\{e_4\}$& $\emptyset$\\
      \hline 
     \end{tabular}
  \end{center}
\end{table}

Doing DGE will not bring more pruning in the classical set-intersection based algorithm, and there will be additional overhead when updating edges. It also needs to be combined with other optimization methods. In the next subsection, we will give a new algorithm called DGEE combined with Failing Set technique. 

\subsection{DGEE}
In this subsection, we will introduce the DGEE based on DGE combined with Failing Set \cite{han2019efficient} optimization.

Failing set is actually a set of related vertices that can change the current failing state. It defines a closure set of ancestors for $ancestors[u] = \bigcap_{u_n \in (N(u) \cap{\varphi})}{ancestors[u_n]}$. It is strongly related to the edges in the graph. 

DGEE's pseudo code is shown in algorithm \ref{algo:pipeline}. Line 1-4 indicate that the match can be found along this branch of the DFS, so the failing set is the complete set of vertices. Line 6 indicates that when enumerating new vertex's mapping, we are not sure whether there is a solution, so the FS is set to be empty. Lines 7-15 enumerate each candidate vertex and enter the next layer in Line 10. There are two important parts, the Update in Line 9 and the use of FS in Lines 11-14. 

\begin{algorithm}
\SetAlgoLined
\KwResult{The matching results}
 \If{$k=|V(Q)|$ } {
    Report it\;
    $FS[k-1] = V(Q)$\;
    \Return\;
 }
 $FS[k] = \emptyset$\;
 \For{$v \in {VC}_{u_k} $ } {
    $M' = M \cup \{(u, v)\}$\;
    Update(ancestors, $VC_{u_{k+1}}$, $u_{k+1}$)\;
    DGEE(Q, G, k+1)\;
    \If {$ u_k \notin FS[k] $} {
        $FS[k-1]=FS[k]$\;
        \Return\;
    }
 }
 $FS[k-1] = FS[k-1] \cup{FS[k]}$\;

 \caption{DGEE(Q, G, k)}
 \label{algo:pipeline}
\end{algorithm}

Lines 11-14, backtracking after performing a search, we will know whether there is a solution along this search branch, and the FS will also be obtained at this time. Then if current vertex $u_k$ is out of the set, the subsequent search will be wasteful and can backtrack directly.  And Line 16 is to inherit FS back. Another difference between this and FS is that the judgement of the visited vertex is cancelled, and this part has been added to the edited graph in the form of an edge. 

The Update's pseudo code is shown in algorithm \ref{algo:update}. It is a modification of the DGE algorithm. The main difference between it and algorithm \ref{algo:DGE} is that it updates $VC(u)$ at enumerating vertex $u$'s candidate. However, in algorithm \ref{algo:DGE}, all subsequent vertex's valid candidate set are updated at enumerating a vertex $u$. Experimental results show that the running time of the former is shorter. Although their time complexity is the same, the constant of algorithm \ref{algo:update} is smaller. 

Line 1 is the initialization of VC(u). Usually when doing $C_u^{u',v'} \cap VC(u)$ at first time, the size of $VC(u)$ will change. So in the implementation, this part will be directly equal to $C_u^{u',v'} \cap VC(u)$, where $u'$ is the first enumerated neighbor of $u$ ordered by DFS level. Line 2-7 is to initialize the edges and then delete the edges in $E_{del}$. In the optimization of FS, the role of the edge is reflected on ancestors, so only ancestors are needed to modified. Line 9-13 is to add edges in $E_{add}$. 

\begin{algorithm}
\SetAlgoLined
\KwResult{Update implicit edited graph and VC}
 $VC(u) = C(u)$\;
 $ancestors[u] = \emptyset$\;
 \For(){$u_n \in \varphi$} {
    $VC(u) = C_u^{u_n,v_n} \cap VC(u)$\;
    \If{size of $VC_u$ changed} {
        $ancestors[u] = ancestors[u_n] \cup ancestors[u]$\;
    }
 }
 \For(){$v \in VC(u)$} {
    \If{$v$ is mapped by $u'$} {
        $VC(u) =VC(u) - \{v\}$\;
        $ancestors[u] = ancestors[u'] \cup ancestors[u]$\;
    }
 }
 \caption{Update($VC(u)$)}
 \label{algo:update}
\end{algorithm}

In fact, the edited graph obtained by algorithm \ref{algo:DGE} and algorithm \ref{algo:update} may be different, since the order of adding edges process is different. But in practice, we have compared the two methods, and the pruning power is not much different. Since the $VC(u)$ is smaller after all intersection operations completed, the algorithm \ref{algo:update} will take less time. 

\section{Experiment Results}
In this section, we will evaluate the performance of two proposed algorithms, fGQL and DGEE. In the experiment, we compare the overall performance and the performance of each part.
Firstly, we compare the combined two algorithms with the state of the art algorithm, the \textbf{baseline} \cite{sun2020memory}. Secondly, we compare the fGQL with GQL in preprocessing time. Finally, we compare the DGEE with LFTJ, the state of the art enumeration method recommended in paper \cite{sun2020memory}.


\subsection{Experimental setup}
\textbf{Baseline. } It \cite{sun2020memory} points that a faster algorithm can be achieved through a combination of different optimizations of the current algorithms by analyze eight algorithms (QuickSI, GraphQL, CFL, CECI, DAF, RI, VF2++, Glasgow). And we use the settings it recommended, which can achieve best performance: 1) GraphQL as Preprocessing method. 2) RI as reordering method. 3) CECI as enumeration method. and 4) optimization technique failing set \cite{han2019efficient} and QFilter \cite{han2018speeding}. The source code refers to the code published in the paper.

\textbf{Datasets.} Our algorithm is tested by 8 real-world data graphs used in paper \cite{sun2020memory} shown in Table \ref{tab:datasets}. The public datasets contains different graph scales, number of labels. And the query graphs are the same as used in paper \cite{sun2020memory}, which is generated by randomly extracting subgraphs from data graph. It have 1800 query graphs for each data graph. These 1800 query graphs can be divided into 9 different scales graphs. For query graphs that can be solved in a limited time, we call them easy graphs. And for query graphs that cannot be solved in a limited time, we call them hard graphs.

\begin{table}[h]
  \begin{center}
    \caption{Eight data graph features}
    \label{tab:datasets}
    \begin{tabular}{l|c|c|c}
      \hline
      Dataset &  $|V|$ & $|E|$ & $|\Sigma|$\\
      \hline
      hprd & 9,460 & 34,998 & 307\\
      yeast & 3,112 & 12,519 & 71 \\
      dblp & 317,080 & 1,049,866 & 15 \\
      eu2005 & 862,664 & 16,138,468 & 40 \\
      youtube & 1,134,890 & 2,987,624 & 25 \\
      patents & 3,774,768 & 16,518,947 & 20 \\
      human & 9,460 &34,998 & 44 \\
      wordnet & 76,853 & 120,399 & 5 \\
    \end{tabular}
  \end{center}
\end{table}

\textbf{Experiment Environment.} We evaluate code on Intel Xeon Gold 6129 CPU. It have 64 cores and 1024GB RAM running CentOS Linux.

\textbf{Runtime Setup.} Since the number of solutions in some query graphs are exponentially explode, the code is terminated after finding 100,000 matches. Since the algorithms we compared use the same code when reading the graph, we do not count the reading time when counting elapsed time.
There are still some hard query graphs during the test, and the running time for algorithms to run to the solution is hard to estimate, so we can't let the program run all the time.
We set maximum running time limit to 5 minutes for each query graphs. 

\textbf{Metrics.} We use two metrics to evaluate the performance: unsolved number and elapsed time. Running time is an important indicator to measure the performance of an algorithm. 
Unsolved number indicates the number of unsolved query graphs within all query graphs in a data graph. It is another importance evaluation criterion, which indicated the ability to solve difficult cases. 
Additionally, our results are given in terms of average values. 
For those query graphs that are forcibly terminated during operation, we set its elapsed time to 5 minutes, when calculating average elapsed time.

\subsection{Overall Performance}
In this subsection, we use the combination of fGQL and DGEE to compare the overall performance with $baseline$. And in general, the new algorithm surpasses the $baseline$ in all aspects. 

In table \ref{table:1}, for four datasets that can be completely solved, the new algorithm can also solve them completely. Moreover, on the eu2005 dataset, the new algorithm can clear it to solve the query graph that could not be solved in $baseline$. Furthermore, we can reduce the number of hard query graphs for three datasets that previously existed hard query graphs. 

\begin{table}[h!]
  \begin{center}
    \caption{Number of unsolved query graphs on real datasets}
    \label{table:1}
    \begin{tabular}{c|c|c}
      \hline
      \textbf{Datasets} &  \textbf{Baseline} &  \textbf{Our Algorithm}  \\
      \hline
      hprd & 0 &  \textbf{0} \\
      yeast & 0 & \textbf{0} \\
      dblp & 18 & \textbf{4} \\
      eu2005 & 1 & \textbf{0} \\
      youtube & 0 & \textbf{0} \\
      patents & 8 & \textbf{5} \\
      human & 8 &  \textbf{2} \\
      wordnet & 0 & \textbf{0} \\
    \end{tabular}
  \end{center}
\end{table}

From the perspective of the elapsed time, our algorithm also outperforms the $baseline$ on eight datasets. Figure \ref{fig:overview} shows an overview of elapsed time. Regardless of the two datasets that are easier to solve, hprd and yeast, there are significant improvements on the other six datasets. 

\begin{figure}[h]
  \centering
  \includegraphics[width=\linewidth]{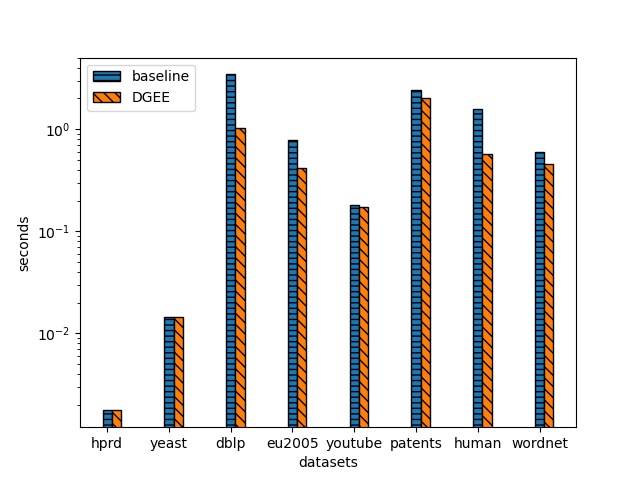}
  \caption{Elapsed time of subgraph matching algorithms on real datasets}
  \Description{Performance on 8 datasets}
  \label{fig:overview}
\end{figure}

For scenarios where the longer the enumeration time is, our algorithm has a stronger power to reduce the execution time. On dblp, the time can be reduced by more than 70\%. For datasets that are not significantly improved like youtube, it is because search space before optimization is small and can be completed soon without pruning. 

Figure \ref{fig:details} is more detailed data to illustrate performance. It is the result of the elapsed time of each scale of graph in each dataset. It can be seen from the figure that our performance is better than $baseline$ on almost all query graphs. Although some easy query graphs will be slower. When the scale of the graph is large, such as 32S and 32D, there is a significant improvement in each dataset. 

And our algorithm works better than $baseline$ in almost graph type, including sparse graph, dense graph, large graph and small graph. And as the scale of the query graph becomes larger, the pruning effect continues to increase, and the performance will be better.  

\begin{figure}[h]
  \centering
  \includegraphics[width=\linewidth]{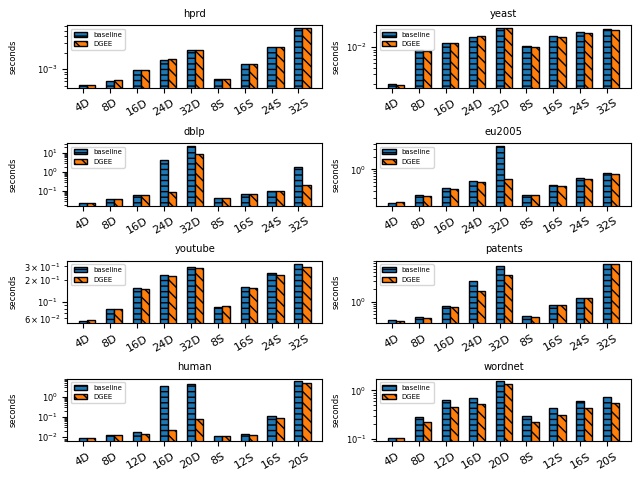}
  \caption{Elapsed time of subgraph matching algorithms on different scale query graphs}
  \Description{Performance on 8 datasets, 9 different graph type}
  \label{fig:details}
\end{figure}

\subsection{Static Pruning performance}

In this subsection, we compare fGQL with CFL, GQL, DAF and CECI in pruning ability. All settings like iteration number are corresponding to original paper's setting. 

The number of average candidate vertices is calculated as following:

\begin{equation}
     avg_{candidate} = \sum_{k=1,..,N} {(\sum_{i=0,...,|V(q)|-1} C(u_i) / |V(q)|)}/N
\end{equation}

Where $N$ is number of query graphs for a data graph.

In table \ref{tab:filter-num}, compared with other algorithms, fGQL has the least number of candidate vertices in each dataset. 
The ratio is the result of fGQL compared with the least number of average candidate vertices among the other algorithms.
The pruning ability of fGQL is significantly improved compared with other algorithms. In comparison with CFL, it can be reduced by more than 40\% in dblp. 

\begin{table}[h!]
  \begin{center}
    \caption{Number of average candidate vertices of five static pruning methods}
    \label{tab:filter-num}
    \begin{tabular}{c|c|c|c|c|c}
      \hline
      &\textbf{CECI} & \textbf{CFL} & \textbf{DAF} &  \textbf{GQL} & \textbf{fGQL}  \\
      \hline
      hp & 4.206 &     3.606&     3.727 & 3.169 & \textbf{3.169}\\
      ye & 48.035&    43.628&    49.176 & 33.868 &    \textbf{33.868} \\
      db & 425.709&   327.834&   272.205 & 188.334 &   \textbf{188.334}\\
      eu & 2179.185&  1985.002&  1816.099 & 1565.373 &  \textbf{1565.373}\\
      yt & 923.2&     837.48&    790.449 & 693.815 &   \textbf{693.815}\\
      up & 3171.876&  2524.463&  1940.351 & 1452.86 &  \textbf{1452.86}\\
      hu & 83.992&    83.378&    88.229 & 79.651 &    \textbf{79.651} \\
      wn & 25301.364& 25270.053& 25624.669 & 19917.665 & \textbf{19917.665}\\
    \end{tabular}
  \end{center}
\end{table}

Due to fewer candidate number, the search space can be reduced, so that more query graphs can be solved in a limited time as in table \ref{tab:unsolved-filter}. Experimental results show that after using fGQL, it can solve more hard query graphs with least number of unsolved query graphs among these algorithms. The outcome of GQL is same to fGQL, so they have 

\begin{table}[h!]
  \begin{center}
    \caption{Number of unsolved query graphs of five static pruning methods}
    \label{tab:unsolved-filter}
    \begin{tabular}{c|c|c|c|c|c}
      \hline
      \textbf{Datasets} & \textbf{CECI} & \textbf{CFL} & \textbf{DAF} & \textbf{GQL} & \textbf{fGQL}  \\
      \hline
      hprd & 0 & 0& 0& 0 &   \textbf{0} \\
      yeast & 0 & 0& 0& 0 &    \textbf{0} \\
      dblp & 21&   21& 23 & 19 &   \textbf{19} \\
      eu2005 & 7 &  6& 10 & 1 &  \textbf{1} \\
      youtube & 0 & 0& 0& 0 &    \textbf{0}\\
      patents & 8 &  8& 8 &8 &  \textbf{8} \\
      human & 12 &    10& 16 & 8 &   \textbf{8} \\
      wordnet & 0 & 0 & 0& 0 & \textbf{0}\\
    \end{tabular}
  \end{center}
\end{table}

Compared with GQL, fGQL has less preprocessing time in figure \ref{fig:filter-time}, such as human and wordnet. 

In addition, in some cases with less redundant operations, such as hprd and dblp, the allocation of extra memory in fGQL will bring extra time overhead, which will increase the time slightly. 


\begin{figure}[h]
  \centering
  \includegraphics[width=\linewidth]{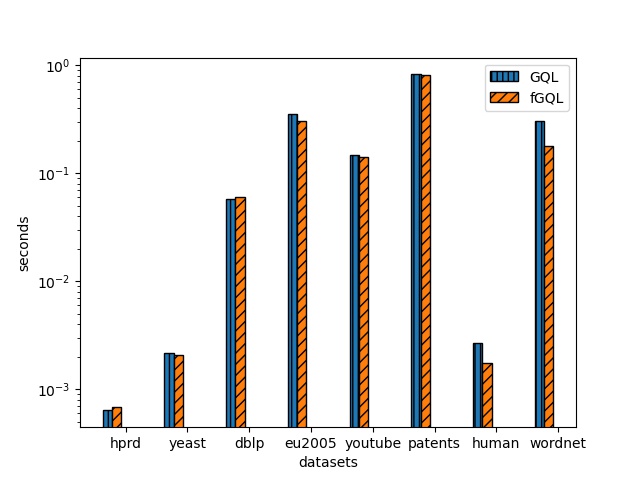}
  \caption{Filter time of GQL and fGQL in eight data graphs}
  \Description{preprocessing time}
  \label{fig:filter-time}
\end{figure}

To sum up, fGQL has strongest pruning ability. And compared with the execution time constraint of 5 minutes, its filter time is very short, less than 1 second, which is 0.33\% of the execution time. 

\subsection{Dynamic Pruning performance}
In this subsection, we analyze the performance of the algorithms after utilizing DGEE. In the above analysis, we get CFL or fGQL is the best choice with the least preprocessing time and the best pruning effect. Therefore, here we compare the combination of DGE with two static filter algorithms.

The size of the search space is a count of the number of nodes in the search tree. The number of it can measure the pruning ability of the algorithm. After using DGEE, the search space can be reduced as shown in table \ref{tab:SpaceSize}. In particular, the larger the search space before optimization, the stronger the dynamic pruning effect in figure \ref{fig:relation}. In the case of using CFL's static pruning method on the eu2005 dataset, the DGEE method can reduce the search space by more than 95\%. 

Since fGQL gets fewer candidate vertices than CFL, the search space of fGQL is smaller than CFL. For the hprd dataset, the effect in fGQL will be worse than that in CFL. The reason is that the set of candidate vertices will affect the deleted edges.

Then it is more interesting that after using DGEE, the performance gap between the different static pruning methods can be narrowed. Some algorithms like CFL with less pruning power but low time cost have stronger power. And it can be seen in figure \ref{fig:relation} that DGEE has a better effect on using CFL as a static pruning method. 100,000 is a dividing line, since we limit the search process until 100,000 solutions are found. Sometimes the solutions will be less than 100,000. 

Moreover, for datasets whose search space is less than 100,000 before optimization, such hprd,yeast and youtube, the solution can be found directly, i.e, there are few invalid branches in searching. Therefore, pruning has little effect in these datasets. And the optimization effect of DGEE is mainly reflected in cases where the search space is greater than 100,000.

\begin{table}[h]
  \begin{center}
    \caption{Size of search space of DGEE/LFTJ with fGQL and CFL static pruning methods}
    \label{tab:SpaceSize}
    \begin{tabular}{c|c|c|c|c}
      \hline
      \textbf{Datasets} &\textbf{fGQL} & \textbf{fGQL+DGEE} & \textbf{CFL} & \textbf{CFL+DGEE} \\
      \hline
hprd &3,012 & 3,012 & 3,015 & \textbf{3,009}  \\
yeast &50,656 & \textbf{50,157} & 56,971 & 54,313  \\
dblp &17,735,657 &  \textbf{2,343,772} & 20,673,852  & 2,348,706  \\
eu2005 &1,141,857 & \textbf{72,474} & 1,789,456 & 78,555  \\
youtube &25,033 & \textbf{24,772} & 25,957 & 25,663  \\
patents &4,457,248 & \textbf{2,260,981} & 4,647,940 & 2,270,049  \\
human &5,811,305 & \textbf{1,580,398} & 6,682,605 & 1,581,332  \\
wordnet &768,023 & \textbf{655,045} & 934,823 & 776,166 \\
    \end{tabular}
  \end{center}
\end{table}

\begin{figure}[h]
  \centering
  \includegraphics[width=\linewidth]{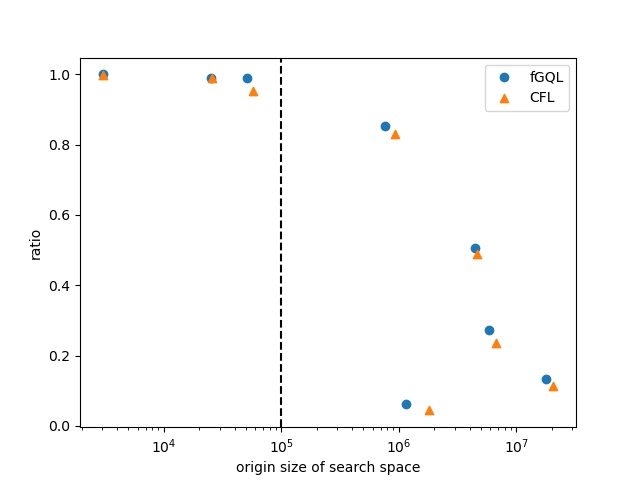}
  \caption{The compression ratio after using DGEE with two static filters}
  \Description{relation}
  \label{fig:relation}
\end{figure}

We compare DGEE with LFTJ. LFTJ is a enumeration method recommended in Sun's \cite{sun2020memory} work, which is the set-intersection based local candidates computation. It takes CECI's enumeration combining QFilter and Failing Set. 

 Figure \ref{fig:DGEforAlgo} shows the performance of DGEE under different static pruning. While reducing the size of the search space, DGEE does not introduce additional enumeration overhead. The results show that whether it is CFL or fGQL, after the introduction of DGEE, the performance on many datasets has been greatly improved, and there is no performance deterioration on all dataset. Moreover, without using DGEE, the performance of CFL anf fGQL are different, and they are applicable to different datasets. But after using DGEE, the running time of CFL on most datasets is lower than fGQL, expect that it is not as good as fGQL on human dataset. The pruning effect of DGEE has little relevance to the static pruning method. That is to say, the content of dynamic pruning and the content of static pruning do not have much overlap. 

\begin{figure}[h]
  \centering
  \includegraphics[width=\linewidth]{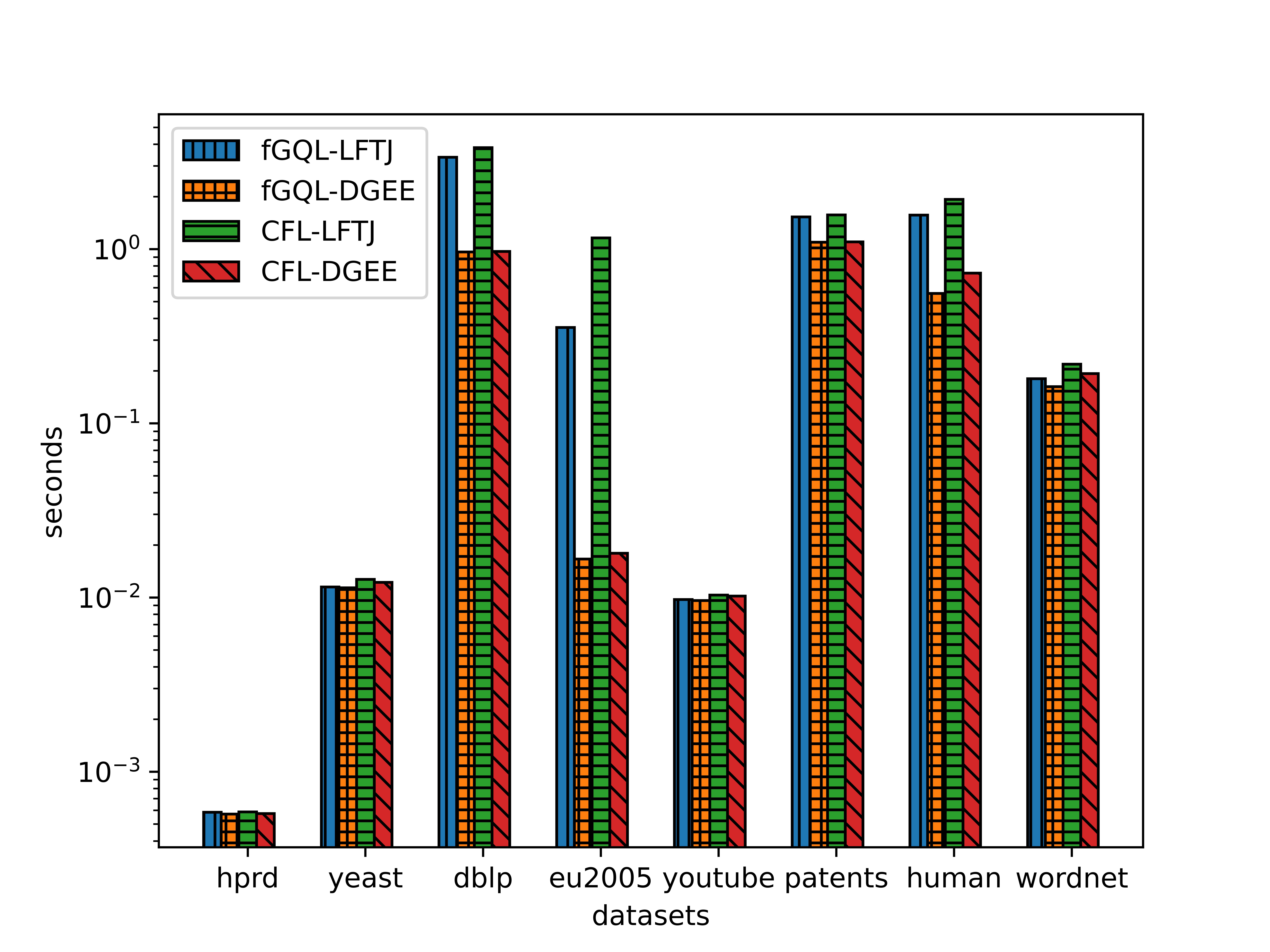}
  \caption{Enumeration time of DGEE/LFTJ with two static pruning methods (CFL and fGQL) on eight data graphs}
  \Description{DGEE time}
  \label{fig:DGEforAlgo}
\end{figure}

Figure \ref{fig:DGEforAlgoDetails} shows that most of the performance improvement is on large query graphs. And it can eliminate the gap caused by static pruning, like 16D in dblp. Additionally, it can be seen that DGEE can improve or maintain the performance regardless of the scale of the query graph. Since more edges can be deleted in the dense query graphs, most of the performance improvement happens in dense query graphs. By the way, there are also a small amount of large sparse query graphs that can be accelerated.

\begin{figure}[h]
  \centering
  \includegraphics[width=\linewidth]{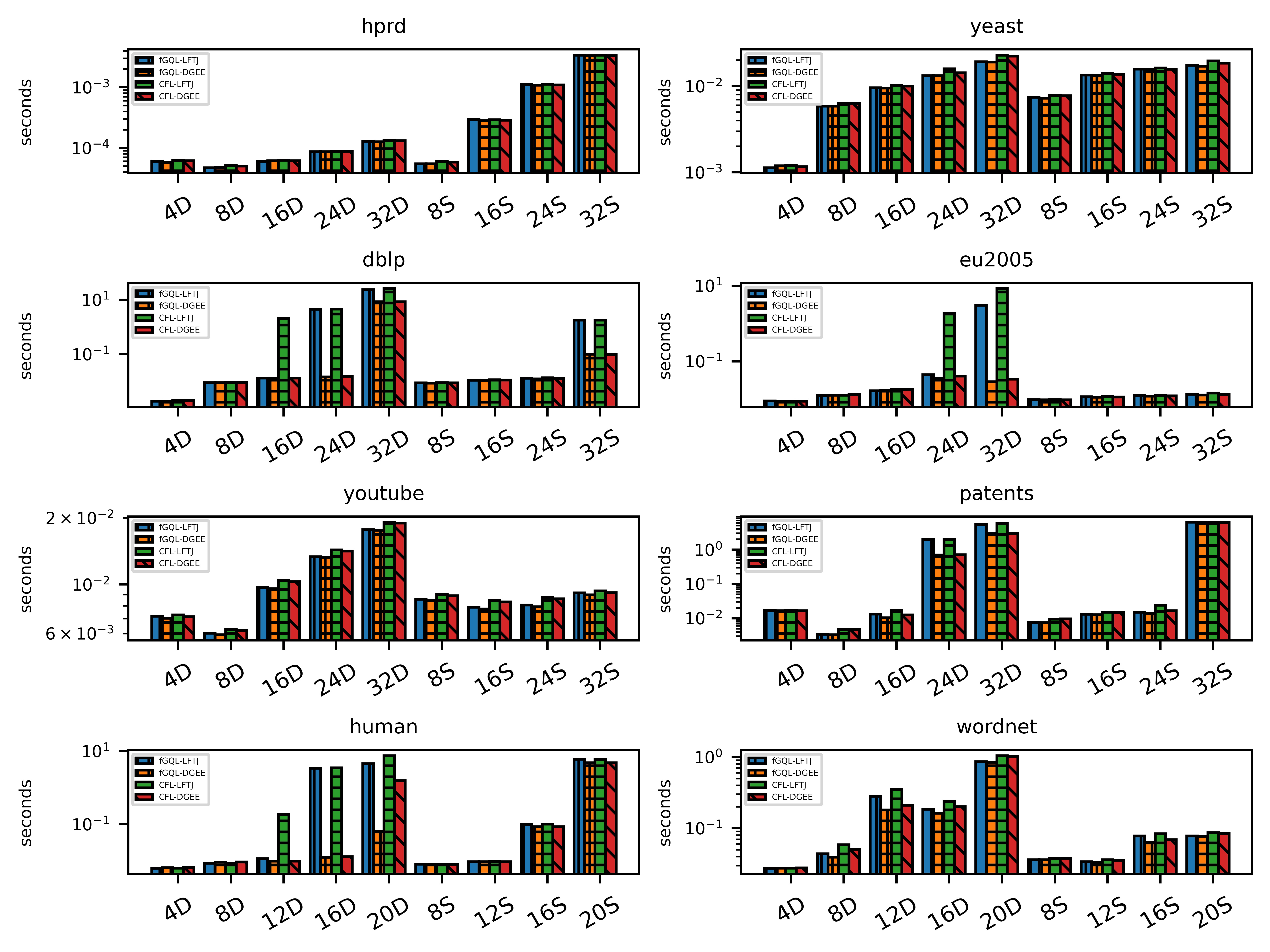}
  \caption{Enumeration time on different scale query graphs}
  \Description{Enumeration time}
  \label{fig:DGEforAlgoDetails}
\end{figure}

To sum up, DGEE can empower existing algorithms. Besides, DGEE does not require a strong static pruning algorithm to achieve high performance and have strong stability in different static pruning methods. Additionally, it is also more suitable for dense query graphs, with several orders of magnitude improvement on some dense query graphs. 

\section{Conclusion}
In this paper, we consider locality concepts in iterative static pruning method, and propose a new static pruning algorithm called fGQL. Compared with the GQL, it can accelerate the preproccessing.
And we propose a Dynamic Graph Editing technique in Assignment under Conditional Candidate Set Problem and develop a new enumeration method, called Dynamic Graph Editing Enumeration. 
Through the comparative analysis of DGEE with different static pruning methods, the pruning effect of DGE has strong stability. 
Extensive experiments on real datasets show that DGEE outperforms the state-of-the-art algorithm both in easy graphs and hard graphs, and can be improved by several orders of magnitude on some cases. 

In addition, applying Dynamic Graph Editing with other algorithms like Versatile Equivalences would be an interesting future work.


\bibliographystyle{ACM-Reference-Format}
\bibliography{sample-base}










\end{document}